\documentclass[conference]{IEEEtran}
\IEEEoverridecommandlockouts

\usepackage[utf8]{inputenc} 
\usepackage[T1]{fontenc}
\usepackage{url}
\usepackage{ifthen}
\usepackage[noadjust]{cite}
\usepackage[cmex10]{amsmath} 
\allowdisplaybreaks
\usepackage{amssymb}
\usepackage{graphicx}
\usepackage{epstopdf}
\usepackage{epsfig}
\usepackage{multirow}
\usepackage{amsthm}
\usepackage{comment}
\usepackage{caption}
\usepackage{enumitem}
\usepackage{amsmath}
\usepackage{color}
\usepackage{bbm}
\usepackage{algorithm}
\usepackage{algpseudocode}
\usepackage{csquotes}
\usepackage{subfiles}
\usepackage{caption, multirow, makecell}
\usepackage{array}
\usepackage{caption}
\usepackage{nicematrix}
\usepackage{cuted}
\usepackage{subcaption}
\usepackage{mathtools}
\usepackage{blkarray}
\DeclarePairedDelimiter\ceil{\lceil}{\rceil}

\usepackage[colorlinks=true, linkcolor=black, citecolor=black, urlcolor=black]{hyperref}

\newtheorem{thm}{Theorem}

\newtheorem{cor}{Corollary}
\newtheorem{example}{Example}

\newtheorem{rem}{Remark}

\allowdisplaybreaks

\def\BibTeX{{\rm B\kern-.05em{\sc i\kern-.025em b}\kern-.08em
    T\kern-.1667em\lower.7ex\hbox{E}\kern-.125emX}}

\begin{document}
\title{Fundamental Limits of Multi-User Distributed Computing of Linearly Separable Functions }
\author{\IEEEauthorblockN{K.~K.~Krishnan Namboodiri, Elizabath Peter, Derya Malak, and Petros Elia\\}
	\IEEEauthorblockA{ Communication Systems Department, EURECOM, Sophia-Antipolis, France \\
		E-mail: \{karakkad, peter, malak, elia\}@eurecom.fr}
       \thanks{This research was partially supported by European Research Council ERC-StG Project SENSIBILITE under Grant 101077361, by the Huawei France-Funded Chair Toward Future Wireless Networks, and in part by the Program ``PEPR Networks of the Future” of France 2030.}
}
\maketitle
\begin{abstract}
This work establishes the fundamental limits of the classical problem of multi-user distributed computing of linearly separable functions. In particular, we consider a distributed computing setting involving $L$ users, each requesting a linearly separable function over $K$ basis subfunctions from a master node, who is assisted by $N$ distributed servers. At the core of this problem lies a fundamental tradeoff between communication and computation: each server can compute up to $M$ subfunctions, and each server can communicate linear combinations of their locally computed subfunctions outputs to at most $\Delta$ users. The objective is to design a distributed computing scheme that reduces the communication cost (total amount of data from servers to users), and towards this, for any given $K$, $L$, $M$, and $\Delta$, we propose a distributed computing scheme that jointly designs the task assignment and transmissions, and shows that the scheme achieves optimal performance in the real field under various conditions using a novel converse. We also characterize the performance of the scheme in the finite field using another converse based on counting arguments.
\end{abstract}


\section{Introduction}
Distributed computing systems are imperative for handling the computationally intensive data-driven tasks that arise in many modern applications. This reality has brought to the fore various distributed computing frameworks, such as MapReduce~\cite{DeG} and Spark \cite{ZCFSS}, that can successfully parallelize computations across clusters of computing nodes. The efficiency of distributed computing is naturally a function --- among other things --- of the network topology, and of the nature of the requested tasks. A very common setting --- which is what we focus on here --- considers multiple servers, multiple clients/users, and linearly-separable tasks/functions. In particular, we consider a setting with $N$ servers and $L$ users, where each user $\ell \in [L]$ wishes to compute a distinct function $F_{\ell}(.)$ that depends on a set of $K$ files $\mathcal{W}=\{W_1, W_2, \ldots, W_K\}$ stored at a master node. As one might expect, the functions $F_{\ell}(\mathcal{W})$, $\ell \in [L]$, can be written as the linear combination of $K$ distinct subfunctions $f_k(W_k)$, $k \in [K]$, and the coefficients are captured by a demand matrix $\mathbf{D}$ of dimension ${L \times K}$. We here assume that the master node coordinates $N$ servers, and that each server can perform at most $M \leq K$ subfunction computations, subject to an additional constraint on the maximum number of users that each serve concurrently. We focus on the practical scenario where each server can communicate messages that are linear functions of their locally computed subfunction outputs, and can do so to a carefully chosen set of up to $\Delta \leq L$ users. No communication among servers is allowed. 
In this context, our objective is to design a distributed computing scheme that achieves the minimum communication cost (minimum rate, i.e., minimum amount of information transmitted from servers to users) for a given set of demands under the computational constraint $M$ and the connectivity constraint $\Delta$. The computation–communication tradeoff has been extensively investigated across various distributed computing frameworks \cite{LMYA, YaYW, NPME, YAb, TLDK, KhE, KhE2, MSSE, BrEl, PNR, PLE, YMA, YLRKSA, WSJC, WSJC2, LLPPR, DFHJCG}. \nocite{MaE, KhE3}
\subsection{Related Works}
Our considered setting is closest to the distributed computing model in \cite{KhE}, which employs tesselation-based tilings to reduce the number of servers, as well as the work in 
\cite{KhE2}, which considered ${F}_{\ell}(\mathcal{W})$'s over a finite field and which designed schemes based on covering codes. Connections can also be found with the work in~\cite{NPME}, which addressed a single-user variant of our setting, and which provided bounds derived from matrix-entropic inequalities and covering designs. 
Naturally, links can also be established with the works in \cite{WSJC, WSJC2} which again study the single-user linearly separable function computation problem in the presence of stragglers, as well as with the work in \cite{TaM} which considered the distributed computation of a linearly separable Boolean function requested by a single user.

\subsection{Our Contributions}
In this work, we study the multi-user distributed computing problem of linearly separable functions under server–user connectivity and computational constraints. Our contributions are summarized as follows.
\begin{itemize}
    \item We first propose an achievability scheme (Theorem~\ref{thm1}) for the multi-user setting based on a partitioning of the demand matrix into suitable subblocks. For each such subblock, we apply a left nullspace–based construction that jointly designs the dataset assignment and the server transmissions to guarantee exact decodability. This nullspace-based design builds on ideas introduced in our recent work \cite{NPME}, which studied the single-user version of the distributed linearly separable function computation problem. In contrast to \cite{NPME}, which is restricted to a single-user setting, the present work extends the framework to the multi-user setting, while retaining validity over general system parameters and arbitrary fields.
    \item We then derive a converse bound on the communication cost that is applicable for computing over finite fields via an equivalent matrix factorization formulation of the problem (Theorem \ref{thm:lb}). In this formulation, the demand matrix is represented as a product $\mathbf{D} = \mathbf{C}\mathbf{A}$, where the sparsity patterns of the factor matrices \(\mathbf{C}\) and \(\mathbf{A}\) are dictated by the $\Delta$-connectivity constraint and the $M$-computing capability constraint of the servers, respectively. The inner dimension of this factorization corresponds to the communication cost, and the converse is obtained by lower bounding the minimum achievable inner dimension over all such factorizations. Using non-trivial counting arguments, we establish a novel converse on the communication cost. When specialized to the single-user setting considered in~\cite{NPME}, the proposed converse is strictly stronger than the previously known bound for every \(q>2\). Moreover, under certain divisibility conditions on $K$, $L$, $M$, and $\Delta$, we show that the achievable communication cost meets the converse as the field size $q \to \infty$. For finite $q \geq K$, we show that our scheme is within a factor of $3$ from the optimal rate when $\Delta$ divides $L$ and $(\Delta+M-1)$ divides $K$, and within a factor of $8$ otherwise (Theorem~\ref{thm:optgap}).
    \item For real fields, we derive an additional converse bound using the same matrix factorization representation of the problem (Theorem~\ref{thm:conv_real}). In this setting, we count the degrees of freedom involved in factorizing the demand matrix as $\mathbf{D} = \mathbf{C}\mathbf{A}$ to lower bound the inner dimension, and hence the communication cost. When $\Delta$ divides $L$ and $(\Delta+M-1)$ divides $K$, we show that the achievable rate in Theorem \ref{thm1} matches this converse bound, and is within a factor of $4$ otherwise.

    \end{itemize}

{\bf Notations.} 
For $m, n \in \mathbb{Z}^+$, we let $[n] = \{1,2,\ldots,n \}$ and $m | n $ denotes $m$ divides $n$. The number of non-zero entries in a vector $\mathbf{x}$ is denoted by $||\mathbf{x}||_0$. A vector $\mathbf{x}$ is said to be $k-$sparse if $||\mathbf{x}||_0 \leq k$. The cardinality of a set $\mathcal{S}$ is denoted by $|\mathcal{S}|$. For any $x \in \mathbb{R}^+$, $\ceil{x}$ denotes the smallest integer greater than or equal to $x$. The vertical concatenation of two matrices $\mathbf{A}_{m_1 \times n}$ and $\mathbf{B}_{m_2 \times n}$ are denoted by $[\mathbf{A};\mathbf{B}]_{(m_1+m_2) \times n}$. The $i$-th row and the $j$-th column of a matrix $\mathbf{A}_{m \times n}$ are denoted by $\mathbf{A}(i,:)$ and $\mathbf{A}(:,j)$, respectively.

\section{System Model}

\begin{figure}
\begin{center}   
\includegraphics[width=0.9\linewidth]{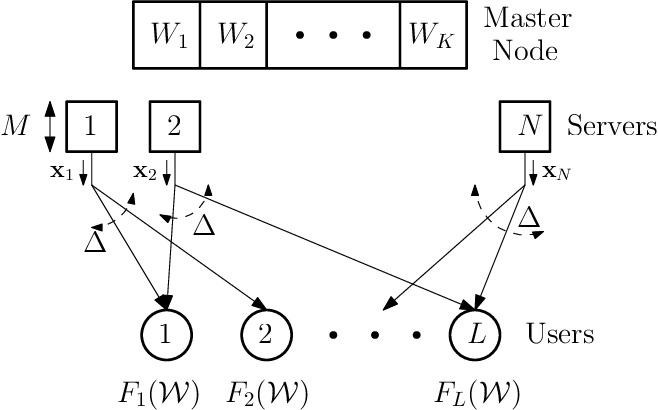}
\caption{A multi-user distributed computing system.}
\label{SystemModel}
\end{center}
\end{figure}

Consider a distributed computing system with $L$ users, each of whom wants to compute a function that depends on a library $\mathcal{W}=\{W_1, W_2, \ldots, W_K\}$ stored at a master node, where $W_i \in \mathbb{F}^B$ for some field $\mathbb{F}$ and $B \in \mathbb{Z}^+$. The requested functions $F_{\ell}(\mathcal{W})$, $\ell \in [L]$, admit linear separability over $K$ basis subfunctions $f_k(W_k)$, $k \in [K]$, as 
\begin{equation}
 F_{\ell}(\mathcal{W}) = \sum_{k=1}^{K}d_{\ell,k}f_k(W_k),\quad \forall \ell \in [L]
 \label{eq:sep}
\end{equation}
\noindent where $d_{\ell,k} \in \mathbb{F}$ and $f_k(W_k)$ can be a linear or a non-linear function. Let $\mathbf{D}=[d_{\ell,k}], {\ell \in [L], k \in [K]},$ be the matrix formed by the coefficients $d_{\ell,k}$ in \eqref{eq:sep}. Then, the $L$ users' demands can be represented as $\mathbf{F}(\mathcal{W})=\mathbf{D}\mathbf{f}$, where $\mathbf{F}(\mathcal{W})=[ F_1(\mathcal{W}), \ldots, F_L(\mathcal{W})]^{\intercal}$ and $\mathbf{f} = [f_1(W_1), \ldots, f_K(W_K)]^{\intercal}$. Since $F_{\ell}(\mathcal{W})$ is linearly separable, the computation of $L$ functions is distributed to $N$ servers with identical computing capabilities. Each server can compute at most $M \leq K$ subfunctions, and can serve at most $\Delta \leq L$ users. Figure~\ref{SystemModel} illustrates the $L$-user distributed computing system described above. The system consists of three phases: \textit{demand phase}, \textit{computing phase}, and \textit{communication phase}.

\textit{Demand phase:} In this phase, each user $\ell \in [L]$ communicates its demand $F_{\ell}(\mathcal{W})$ to the  master node.

 \textit{Computing phase:} Once the demands of all users are known, the master node assigns each server a set of subfunctions to be computed, subject to its computational capacity $M$ (also referred to as computational cost). Let $\mathcal{M}_n \subseteq [K]$ denote the indices of the subfunctions computed by server $n$, where $|\mathcal{M}_n| \leq M$. Then, server $n$ evaluates all those $f_k(W_k)$'s such that $ \mathcal{M}_n \ni k$. The set $\mathcal{M} = \{\mathcal{M}_1, \mathcal{M}_2, \ldots, \mathcal{M}_N \}$ represents the computational tasks assigned to all $N$ servers in the system.

\textit{Communication phase:} In this phase, each server linearly encodes its subfunction values and transmits $r_n$ messages to a subset of users. We consider linear encoding schemes that do not subpacketize the subfunction outputs. The message sent by server $n$ is denoted by
$x_{n,r}, r \in [r_n],$ and is of the form
\begin{align*}
x_{n,r} = \sum_{k \in \mathcal{M}_n}\alpha_{n,r,k}f_k(W_k)
\end{align*}
where $\alpha_{n,r,k} \in \mathbb{F}$. The message $\mathbf{x}_n = [x_{n,1}, \ldots, x_{n,r_n}]$, $n \in [N]$, is intended for at most $\Delta \leq L$ users. Let $\mathbf{x}$ denote the set of transmissions from all $N$ servers, and let $R=\sum_{n \in [N]}r_n$. Then,  $\mathbf{x}$ can be written as
{
\begin{equation*}
\mathbf{x} = \underbrace{\begin{bmatrix}
\alpha_{1,1,1} & \alpha_{1,1,2} & \ldots & \alpha_{1,1,K} \\
\vdots & \vdots & \ddots & \vdots \\
\alpha_{N,r_N,1} & \alpha_{N,r_N,2} & \ldots & \alpha_{N,r_N,K}\end{bmatrix}}_{\mathbf{A} }
\begin{bmatrix}
& f_1(W_1) \\
& \vdots \\
& f_K(W_K)
\end{bmatrix}
\end{equation*}}

\noindent where $\mathbf{A}$ is of dimension $R \times K$ and $||\mathbf{A}(r,:)||_0 \leq M$, $\forall r \in [R]$. The communication links between servers and users are assumed to be non-interfering and error-free. The users linearly combine the received messages from different servers to recover the requested functions. The decoding operation at the users can be represented as $\mathbf{F}(\mathcal{W})=\mathbf{C}\mathbf{x}$, where $\mathbf{C}_{L \times R}$ is composed of the coefficients used to recover the requested functions from the transmissions. Note that the columns of $\mathbf{C}$ have a $\Delta-$sparsity due to the communication constraint at each server. 
 Therefore, 
{
\begin{align*}
  \begin{bmatrix}
    F_1(\mathcal{W}) \\
    \vdots \\
    F_L(\mathcal{W})
 \end{bmatrix}& =
  \begin{bmatrix}
  c_{1,1} & c_{1,2} & \ldots & c_{1,R} \\
  \vdots & \vdots & \ddots & \vdots \\
   c_{L,1} & c_{L,2} & \ldots & c_{L,R}
  \end{bmatrix}
  \mathbf{x}
\end{align*}}
results in the following decomposition
\begin{equation}
\mathbf{D}=\mathbf{C}\mathbf{A}
\label{eq:decomp}
\end{equation}
where $||\mathbf{C}(:,r)||_0 \leq \Delta$, $||\mathbf{A}(r,:)||_0 \leq M$, $\forall r \in [R]$. Equation \eqref{eq:decomp} follows from $\mathbf{F}(\mathcal{W}) = \mathbf{D}\mathbf{f}$ and $\mathbf{x} = \mathbf{A}\mathbf{f}$. Hence, the design of a distributed computing scheme can be viewed as the joint design of two matrices $\mathbf{C}_{L \times R}$ and $\mathbf{A}_{R \times K}$ with the objective of minimizing $R$ under $\Delta-$sparsity and $M-$sparsity constraints on columns of  $\mathbf{C}$ and rows of $\mathbf{A}$, respectively.

One of the performance measures of a distributed computing scheme is communication cost, which is defined as the total amount of data transmitted by servers to users during the communication phase, normalized by the size of a message (assuming messages are of unit size). The communication cost required for a given demand matrix $\mathbf{D}$ under the computational cost $M$ and per-node serving capacity $\Delta$ is denoted by $R_{\mathbf{D}}(M, \Delta)$. Then, $R_{\mathbf{D}}(M, \Delta)$ is obtained as 
\begin{equation}
 R_{\mathbf{D}}(M,\Delta) = \sum_{n \in [N]}r_n.
\end{equation}
Our interest is in the worst-case communication cost, denoted by $R(K,L,M,\Delta)$, which is defined as 
\begin{equation*}
R(K,L,M,\Delta) = \max_{\mathbf{D}}R_{\mathbf{D}}(M,\Delta)
\end{equation*}
over all possible $\mathbf{D}$'s of dimension $L\times K$. The worst-case communication cost is referred to as the rate hereafter. Therefore, the optimal rate, denoted by $R^{*}(K,L,M,\Delta)$, is defined as 
\begin{align*}
R^{*}(K,L,M, \Delta) = \inf \{R(K,L,M,\Delta)&:  R(K,L,M,\Delta) \\ &\text{ is achievable}\}
\end{align*}
where the infimum is taken over all possible task assignments and linear transmission and decoding policies satisfying the computational cost $M$ and the communication constraint $\Delta$. 
We aim to identify the optimal communication-computation costs for a given distributed computing system and design a distributed computing scheme that achieves the optimal rate for a given $K,L,M,$ and $\Delta$.

\section{Achievability Results}
In this section, we present an achievable scheme for the multi-user distributed
linearly separable function computation problem that applies to all values of
$K$, $L$, $M$, and $\Delta$.
\begin{thm}
    \label{thm1}
    For the $(K,L,M,\Delta)$ distributed linearly separable function computation problem, the rate
    \begin{equation}
    \label{eq:thm1}
        R_{\text{ach}}(K,L,M,\Delta) = \Delta\left\lceil\frac{L}{\Delta}\right\rceil\left\lceil{\frac{K}{\Delta+M-1}}\right\rceil
    \end{equation}
    is achievable.   
\end{thm}
\begin{proof}

Consider \(\mathbf{D}\in\mathbb{F}^{L\times K}\). We partition \(\mathbf{D}\) into submatrices
\(\mathbf{D}_{i,j}\) with 
\(i\in\left[\left\lceil \frac{L}{\Delta}\right\rceil\right]\) and 
\(j\in\left[\left\lceil \frac{K}{\Delta+M-1}\right\rceil\right]\) as
\[
\mathbf{D}=
\begin{bmatrix}
\mathbf{D}_{1,1}&\cdots&\mathbf{D}_{1,\left\lceil \frac{K}{\Delta+M-1}\right\rceil}\\
\vdots&\ddots&\vdots\\
\mathbf{D}_{\left\lceil \frac{L}{\Delta}\right\rceil,1}
&\cdots&
\mathbf{D}_{\left\lceil \frac{L}{\Delta}\right\rceil,\left\lceil \frac{K}{\Delta+M-1}\right\rceil}
\end{bmatrix}.
\]
Each block \(\mathbf{D}_{i,j}\) has dimension
\(
\min\{\Delta,\,L-\Delta(i-1)\}\times
\min\{\Delta+M-1,\,K-(\Delta+M-1)(j-1)\}.
\)

    We now design the computing phase and the communication phase for the given demand matrix. 
    First, we consider a sub-demand matrix \(\mathbf{D}_{i,j}\). Corresponding to \(\mathbf{D}_{i,j}\), we require \(\Delta\) servers\footnote{The total number of servers required is \(N=\Delta\left\lceil\frac{L}{\Delta}\right\rceil\left\lceil\frac{K}{\Delta+M-1}\right\rceil\).}. We refer to them as the servers in {group \((i,j)\)}. 
    The set of column indices of \(\mathbf{D}_{i,j}\) is denoted by \(\mathcal{K}_j \triangleq \{(j-1)(\Delta+M-1)+1,(j-1)(\Delta+M-1)+2,\dots,\min(j(\Delta+M-1),K)\}\). Note that \(\mathcal{K}_j\) is also the indices of the files handled by the servers in group \((i,j)\). Without loss of generality\footnote{If $\mathrm{rank}(\mathbf{D}_{i,j}) < \Delta$—which may occur either due to linearly dependent rows in $\mathbf{D}_{i,j}$ or because $\mathbf{D}_{i,j}$ has fewer than $\Delta$ rows when $i = \left\lceil \frac{L}{\Delta} \right\rceil$—the rank can be made equal to $\Delta$ by replacing dependent rows or augmenting the matrix with additional linearly independent rows. Furthermore, the removed rows can be recovered as linear combinations of the remaining rows. If the number of columns itself is less than \(\Delta\), for \(j=\left\lceil\frac{K}{\Delta+M-1}\right\rceil\), the case is trivial and will be dealt with later in this proof. }, we assume that \(\mathrm{rank}(\mathbf{D}_{i,j})=\Delta\). Let \(\mathcal{Y}_{i,j} =\{y_1^{(i,j)},y_2^{(i,j)},\dots,y_\Delta^{(i,j)}\}\subseteq \mathcal{K}_j\) be the indices of the columns that form a \(\Delta \times \Delta\) full-rank submatrix of \(\mathbf{D}_{i,j}\). 
    Then, we form the task assignment set of the \(\Delta\) servers corresponding to \(\mathbf{D}_{i,j}\) as
     \begin{equation}
     \label{eq:dataset1}
         \mathcal{M}_{i,j} = \{\mathcal{M}_1^{(i,j)},\mathcal{M}_2^{(i,j)},\dots,\mathcal{M}_\Delta^{(i,j)}\} 
     \end{equation}
     where, for every \(\delta\in [\Delta]\)
     \begin{equation}
         \label{eq:datasetassign}
         \mathcal{M}_\delta^{(i,j)} = \left\{y_\delta^{(i,j)}\right\} \bigcup \left(\mathcal{K}_j\backslash \mathcal{Y}_{i,j}\right).
     \end{equation} 
     That is, the \(\delta\)-th server in the group \((i,j)\) is assigned to compute the subfunctions \(f_k(W_k)\) for every \(k\in \mathcal{M}_\delta^{(i,j)}\). Each server in group \((i,j)\) has access to all files indexed by \(\mathcal{K}_j\setminus\mathcal{Y}_{i,j}\). Moreover, within group \((i,j)\), the file indexed by \(y_\delta^{(i,j)}\) is uniquely assigned to the \(\delta\)-th server. From \eqref{eq:datasetassign}, it is straightforward to see that \(|\mathcal{M}_\delta^{(i,j)}|\leq M\) for every possible \(i,j,\) and \(\delta\). At this point, note that for a fixed $i$, all servers\footnote{Observe that there is a one-to-one correspondence between servers and task assignments $\mathcal{M}_\delta^{(i,j)}$.} corresponding to $j \in \left[\left\lceil \frac{K}{\Delta + M - 1} \right\rceil\right]$ and $\delta \leq \Delta$ are connected to the users indexed from $(i-1)\Delta + 1$ to $\min(i\Delta, L)$. This is consistent with the constraint that each server can communicate with at most \(\Delta\) users.

     We now design the message transmitted by the \(\delta\)-th server in the group \((i,j)\). Note that this message is intended for users indexed from $(i-1)\Delta + 1$ to $\min(i\Delta, L)$. In order to design the transmission, consider the submatrix of \(\mathbf{D}_{i,j}\), denoted with \(\mathbf{D}_{i,j}^\delta\), formed by removing the columns indexed by \(\mathcal{M}_\delta^{(i,j)}\). Observe that the size of the matrix \(\mathbf{D}_{i,j}^\delta\) is \(\Delta \times (\Delta-1)\) and \(\mathrm{rank}(\mathbf{D}_{i,j}^\delta)=\Delta-1\). This is a consequence of the fact that \(\mathbf{D}_{i,j}^\delta\) consists of \(\Delta-1\) linearly independent columns from the originally chosen \(\Delta\) linearly independent columns of \(\mathbf{D}_{i,j}\). Now, we find a non-trivial vector \(\boldsymbol{\nu}_{i,j}^\delta \in \mathbb{F}^{1 \times \Delta}\) in the left nullspace of \(\mathbf{D}_{i,j}^\delta\). Then the transmission from the \(\delta\)-th server in the group \((i,j)\) is
     \begin{equation}
         \label{eq:trans}
        x_{i,j}^\delta =  \boldsymbol{\nu}_{i,j}^\delta\mathbf{D}_{i,j}[f_{(j-1)(\Delta+M-1)+1},
    \dots, f_{\min(j(\Delta+M-1),K)}]^\intercal.
     \end{equation}
    Note that $\boldsymbol{\nu}_{i,j}^\delta$ is designed such that
$\boldsymbol{\nu}_{i,j}^\delta\mathbf{D}_{i,j}^{\delta} = \mathbf{0}$. Consequently, the support of the vector $\boldsymbol{\nu}_{i,j}^{\delta}\mathbf{D}_{i,j}$ is precisely the set $\mathcal{M}_{\delta}^{(i,j)}$, which enables the $\delta$-th server to construct the transmission message in \eqref{eq:trans}. Corresponding to each sub-demand matrix there are \(\Delta\) computing servers, each making a single transmission of subfunction output size\footnote{If $\mathbf{D}_{i,\left\lceil \frac{K}{\Delta+M-1} \right\rceil}$ has fewer than $\Delta$ columns, the number of computing servers can be chosen to equal the number of columns. Each server computes the subfunction associated with a distinct column index of $\mathbf{D}_{i,\left\lceil \frac{K}{\Delta+M-1} \right\rceil}$ and transmits the corresponding output to the connected users, yielding a communication cost equal to the number of columns, which is upper bounded by $\Delta$.}. Therefore, the rate achieved by the scheme is 
     \begin{equation*}
        R_{\text{ach}}(K,L,M,\Delta) = \Delta\left\lceil\frac{L}{\Delta}\right\rceil\left\lceil{\frac{K}{\Delta+M-1}}\right\rceil.
    \end{equation*}
We now show that each user can decode its requested function. Consider user \(\ell\), where \(\ell \in [(i'-1)\Delta+1:\min(i'\Delta,L)]\) for some \(i'\in \left[\left\lceil\frac{L}{\Delta}\right\rceil\right]\). Alternatively, we have \(\ell = (i'-1)\Delta+\ell_{i'}\) for some \(i'\in \left[\left\lceil\frac{L}{\Delta}\right\rceil\right]\) and \(\ell_{i'}\leq \Delta\). Then, for \(\ell\in [L]\), the requested function of the \(\ell\)-th user can be expressed as 
\begin{align}
    &F_\ell(\mathcal{W})=\mathbf{D}(\ell,:)\mathbf{f}= \sum_{j=1}^{\left\lceil\frac{K}{\Delta+M-1}\right\rceil}\mathbf{D}_{i',j}(\ell_{i'},:)[f_{(j-1)(\Delta+M-1)+1},\notag\\&
     \qquad\qquad\qquad\qquad\qquad\qquad \dots, f_{\min(j(\Delta+M-1),K)}]^\intercal\label{eq:decodesum}.
\end{align}  

Recall that user \(\ell\) receives transmissions from all servers in groups \((i',j)\), for all \(j \in \left[\left\lceil \frac{K}{\Delta + M - 1} \right\rceil\right]\). Therefore, user \(\ell\) gets \(x_{i',j}^\delta\) for every \(j \in \left[\left\lceil \frac{K}{\Delta + M - 1} \right\rceil\right]\) and \(\delta\leq \Delta\). For every \(j \in \left[\left\lceil \frac{K}{\Delta + M - 1} \right\rceil\right]\), user \(\ell\) has
    \begin{equation}
        \label{eq:decode1}
        \mathbf{X}_{i',j} = \mathbf{V}_{i',j}\mathbf{D}_{i',j}[f_{(j-1)(\Delta+M-1)+1},
    \dots, f_{\min(j(\Delta+M-1),K)}]^\intercal
    \end{equation}
     where \(\mathbf{X}_{i',j} \triangleq [x_{i',j}^1;x_{i',j}^2;\dots;x_{i',j}^\Delta]\) and \(\mathbf{V}_{i',j}\triangleq [\boldsymbol{\nu}_{i',j}^1;\boldsymbol{\nu}_{i',j}^2;\dots;\boldsymbol{\nu}_{i',j}^\Delta]\). 
     
     We now show that the matrix \(\mathbf{V}_{i,j}\) is invertible for any \(i\) and \(j\). In other words, we show that \(\mathrm{rank}(\mathbf{V}_{i,j} )=\Delta\), for every \(i\) and \(j\). Suppose \(\mathrm{rank}(\mathbf{V}_{i,j})<\Delta\), then there exists a row \(\boldsymbol{\nu}_{i,j}^{\delta'}\) of \(\mathbf{V}_{i,j}\) which can be represented as 
     \begin{equation}
     \label{eq:contra0}
         \boldsymbol{\nu}_{i,j}^{\delta'} = \sum_{\delta=1,\delta\neq \delta'}^\Delta \alpha_\delta \boldsymbol{\nu}_{i,j}^{\delta}
     \end{equation}
     where \(\alpha_\delta\in \mathbb{F}\) and \(\alpha_\delta \neq 0\) for some \(\delta\). On one hand, \(\boldsymbol{\nu}_{i,j}^{\delta'}\) is a vector in the left nullspace of \(\mathbf{D}_{i,j}^{\delta'} \), and thus
     \begin{equation}
         \label{eq:contra1}
          \boldsymbol{\nu}_{i,j}^{\delta'}\mathbf{D}_{i,j,\mathcal{Y}} = \beta_{\delta'}\mathbf{e}_{\delta'}^\intercal
     \end{equation}
     where \(\mathbf{D}_{i,j,\mathcal{Y}}\) is the submatrix of \(\mathbf{D}_{i,j}\) obtained by choosing only the columns indexed with the set \(\mathcal{Y}_{i,j}\), and \(\beta_{\delta'}\in \mathbb{F}\backslash \{0\}\) and \(\mathbf{e}_{\delta'}\in \mathbb{F}^{\Delta \times 1}\) is an \(\Delta\)-length vector with a \(1\) in the \(\delta'\)-th position and \(0\)-s elsewhere. 
     On the other hand, from \eqref{eq:contra0}, we have
     \begin{align}
         \boldsymbol{\nu}_{i,j}^{\delta'}&\mathbf{D}_{i,j,\mathcal{Y}} = \left(\sum_{\delta=1,\delta\neq \delta'}^\Delta \alpha_\delta \boldsymbol{\nu}_{i,j}^{\delta}\right)\mathbf{D}_{i,j,\mathcal{Y}}\notag\\
         & = \sum_{\delta=1,\delta\neq \delta'}^\Delta \alpha_\delta \bigg(\boldsymbol{\nu}_{i,j}^{\delta}\mathbf{D}_{i,j,\mathcal{Y}}\bigg)= \sum_{\delta=1,\delta\neq \delta'}^\Delta \alpha_\delta\left(\beta_\delta\mathbf{e}_\delta^\intercal\right)\label{eq:contra2}
     \end{align}
     where \(\beta_\delta \in \mathbb{F}\backslash \{0\}\), for every \(\delta\in [\Delta]\backslash \{\delta'\}\). Clearly, \eqref{eq:contra1} and \eqref{eq:contra2} contradict. Therefore, \(\mathbf{V}_{i,j}\) cannot have rank less than \(\Delta\). In other words, for every \(i\) and \(j\), we have \(\mathrm{rank}(\mathbf{V}_{i,j} )=\Delta\).

     Therefore, for every \(j \in \left[\left\lceil \frac{K}{\Delta + M - 1} \right\rceil\right]\), the user can obtain
 \begin{align*}
        \mathbf{D}_{i',j}[f_{(j-1)(\Delta+M-1)+1},
    \dots, f_{\min(j(\Delta+M-1),K)}]^\intercal = \mathbf{V}_{i',j}^{-1}\mathbf{X}_{i',j}.
    \end{align*}
  Consequently, user \(\ell\) obtains \(\mathbf{D}_{i',j}(\ell_{i'},:)[f_{(j-1)(\Delta+M-1)+1},
  \dots, f_{\min(j(\Delta+M-1),K)}]^\intercal\) for every \(j \in \left[\left\lceil \frac{K}{\Delta + M - 1} \right\rceil\right]\). Then, using \eqref{eq:decodesum}, the user can decode its desired function. This completes the proof of Theorem~\ref{thm1}.
\end{proof}
We now present an example of the proposed scheme and demonstrate the achieved rate.
\begin{example}
      Consider the  \((K=10,L=6,M=3,\Delta=3)\) distributed computing system. Let the demand matrix be
      {
      \begin{equation*}
          \mathbf{D}=\begin{bmatrix}
 1 & 1 & 1 & 1 & 1 & 1 & 1 & 1 & 1 & 1 \\
 1 & 2 & 3 & 4 & 5 & 6 & 2 & 8 & 9 & 10 \\
 1 & 4 & 9 & 5 & 3 & 3 & 5 & 8 & 4 & 1 \\
 1 & 3 & 4 & 9 & 4 & 1 & 2 & 6 & 3 & 10 \\
 1 & 5 & 6 & 3 & 9 & 2 & 4 & 4 & 5 & 1 \\
 1 & 6 & 7 & 1 & 9 & 5 & 10 & 1 & 10 & 10
          \end{bmatrix}\in \mathbb{F}_{11}^{6\times 10}.
      \end{equation*}}
      Upon knowing the user demands, the master node partitions \(\mathbf{D}\) into \(\mathbf{D}_{1,1}\), \(\mathbf{D}_{1,2}\), \(\mathbf{D}_{2,1}\), and \(\mathbf{D}_{2,2}\) where each \(\mathbf{D}_{i,j}\) is a matrix of dimension \(3\times 5\). Corresponding to each sub-demand matrix \(\mathbf{D}_{i,j}\), the master node assigns \(3\) computing servers. As a result, the system requires \(N=12\) servers. Servers in groups $(1,1)$ and $(1,2)$ are connected to users $1$, $2$, and $3$, and servers in groups $(2,1)$ and $(2,2)$ are connected to users $3$, $4$, and $5$. Due to space constraints, we describe only the dataset assignment and the transmissions corresponding to the servers in group \((1,1)\). Note that $\mathcal{K}_1 = \{1,2,3,4,5\}$ represents the set of column indices, or equivalently subfunctions, associated with the servers in this group.
      Further, the first three columns of \(\mathbf{D}_{1,1}\) are linearly independent, and thus we have \(\mathcal{Y}_{1,1} = \{1,2,3\}\). From \eqref{eq:datasetassign}, we have the dataset assignment in the servers in group \((1,1)\) as follows:
      \begin{equation*}
          \mathcal{M}_1^{(1,1)}=\{1,4,5\},~\mathcal{M}_2^{(1,1)}=\{2,4,5\},~ \mathcal{M}_3^{(1,1)}=\{3,4,5\}.
      \end{equation*}
      Note that the first server does not have access to \(W_2\) and \(W_3\). Now, we find a vector that resides in the left nullspace of the submatrix \( \mathbf{D}_{1,1}^1\), which comprises the columns \(\{2,3\}\) of \(\mathbf{D}_{1,1}\). Choose a vector, for example, $\boldsymbol{\nu}_{1,1}^{1} = [6,\,-5,\,1]$, from the left nullspace of $\mathbf{D}_{1,1}^{1}$. Then, the transmission made by the first server is
\begin{align}
    {x}_{1,1}^1 &= [6, -5, 1]\mathbf{D}_{1,1}[f_1(W_1),f_2(W_2),\dots,f_5(W_5)]^\intercal\notag\\
    &=2f_1(W_1)-9f_4(W_4)-16f_5(W_5).\label{eq:ex1x1}
\end{align}
 Similarly, the vectors \([3, -4, 1]\) and \([2, -3, 1]\) are in the left nullspaces of \(\mathbf{D}_{1,1}^2 \) and \( \mathbf{D}_{1,1}^3\), respectively. Therefore, the transmissions from server 2 and server 3 in group \((1,1)\) are
\begin{align}
    {x}_{1,1}^2 &= [3, -4, 1]\mathbf{D}_{1,1}[f_1(W_1),f_2(W_2),\dots,f_5(W_5)]^\intercal\notag\\
    &=-f_2(W_2)-8f_4(W_4)-14f_5(W_5) ~ \text{ and}\label{eq:ex1x2}\\
    {x}_{1,1}^3 &= [2, -3, 1]\mathbf{D}_{1,1}[f_1(W_1),f_2(W_2),\dots,f_5(W_5)]^\intercal\notag\\
    &=2f_3(W_3)-5f_4(W_4)-10f_5(W_5)\label{eq:ex1x3}
\end{align}
respectively. Note that users \(1,2,\) and \(3\) have access to \({x}_{1,1}^1\), \({x}_{1,1}^2\), and \({x}_{1,1}^3\). Therefore, from \eqref{eq:ex1x1}, \eqref{eq:ex1x2}, and \eqref{eq:ex1x3}, the users can decode the following 
\begin{align*}
\mathbf{D}_{1,1}(1,:) [f_1(W_1),\dots,f_5(W_5)]^\intercal &=\frac{1}{2}(x_{1,1}^1-2x_{1,1}^2+x_{1,1}^3)\\ 
\mathbf{D}_{1,1}(2,:) [f_1(W_1),\dots,f_5(W_5)]^\intercal &=\frac{1}{2}(x_{1,1}^1-4x_{1,1}^2+3x_{1,1}^3)\\
\mathbf{D}_{1,1}(3,:) [f_1(W_1),\dots,f_5(W_5)]^\intercal &=\frac{1}{2}(x_{1,1}^1-8x_{1,1}^2+9x_{1,1}^3).
\end{align*}
Similarly, user $u$, for $u = 1,2,3$, can compute
$\mathbf{D}_{1,2}(u,:) [f_6(W_6), \dots, f_{10}(W_{10})]^{\intercal}$ from the transmissions sent by the servers in group $(1,2)$. Finally, user $u$, for $u = 1,2,3$, can compute its requested function as
\begin{align*}
    F_u(\mathcal{W}) = \mathbf{D}_{1,1}(u,:)& [f_1(W_1),\dots,f_5(W_5)]^\intercal+\\
    &\mathbf{D}_{1,2}(u,:) [f_6(W_6),\dots,f_{10}(W_{10})]^\intercal.
\end{align*}
Similarly, users $4$, $5$, and $6$ can decode their requested functions from the transmissions from the servers in groups $(2,1)$ and $(2,2)$. Hence, the resulting rate is $R_{\text{ach}} = 12$.

For the same parameters $K=10$, $L=6$, and $\Delta=3$, the scheme presented in \cite{KhE} for lossless recovery of user demands achieves the same rate $R=12$, but requires a computation cost of $M=5$, whereas our scheme requires only $M=3$. This improvement is enabled by breaking away from the disjoint support assumption adopted in \cite{KhE}. It is worth noting, however, that our dataset assignment depends on the user demands, while the scheme in \cite{KhE} allows for a demand-agnostic dataset assignment.

\end{example}
  
\begin{rem}
    \label{rem:scheme}
    The achievability proof builds on the nullspace-based sparsification of the demand matrix proposed in \cite{NPME}, which considers a single-user multiple-function request model and applies the method to sub-blocks of size $L \times (L+M-1)$. In the present multi-user setting, the server-to-user connectivity constraint forces the same approach to be applied to smaller sub-blocks of size $\Delta \times (\Delta+M-1)$.
    
    This construction improves upon the tessellation-based solution in \cite{KhE}, where sub-blocks of size $\Delta \times M$ are chosen without an additional nullspace-based intra-block sparsification. The improvement, however, requires knowledge of the demand matrix, whereas the scheme in \cite{KhE} follows a demand-agnostic task assignment strategy.
\end{rem}

\section{Converse}
In this section, we derive a lower bound on the achievable rate of a $(K,L,M,\Delta)$ distributed computing system. First, we consider the finite field setting where the converse is developed using some counting arguments in the equivalent matrix factorization formulation. Using the derived lower bound, we demonstrate that our scheme approaches optimal performance as $q \rightarrow \infty$. When $q$ is finite, we show that our scheme is order-optimal. Later, we show that our scheme is optimal over $\mathbb{R}$ under certain specific conditions by deriving a matching lower bound using the idea of degrees of freedom in the matrix factorization formulation.

\begin{thm}
For any $(K,L,M,\Delta)$ linearly separable distributed computing problem defined over the finite field $\mathbb{F}_q$, the optimal rate $R^{*}
(K,L,M,\Delta)$ follows
{
\begin{equation}
\begin{aligned}
  R^{*}(K,&L,M,\Delta) \geq  \max\Big (L,\\
 &    \frac{LK}{\Delta + M+\log_q\binom{L}{\Delta}+\log_q\binom{K}{M}-\log_q(q-1)} \Big).
\end{aligned}
\label{eq:conv_lb}
\end{equation}}
\label{thm:lb}
\end{thm}

\begin{proof}
Considering the worst-case scenario, we can assume that users' requests are linearly independent in $f_k(W_k)$'s. Therefore, we have
\begin{equation}
 R^{*}(K,L,M,\Delta) \geq L.
 \label{eq:con_triv}
\end{equation}
Now, let us look at the $(K,L,M,\Delta)$ distributed computing problem through the lens of matrix factorization.  Given the demand matrix $\mathbf{D} \in \mathbb{F}_q^{L \times K}$ of rank $L$, we need to design two matrices $\mathbf{C} \in \mathbb{F}_q^{L \times R}$ and $\mathbf{A} \in \mathbb{F}_q^{R \times K}$, where $R \geq L$, such that $\mathbf{D}=\mathbf{C}\mathbf{A}$ subject to $||\mathbf{C}(:,r)||_0 \leq \Delta$ and $||\mathbf{A}(r,:)||_0 \leq M$, $\forall r \in [R]$. Given $L$ and $K$, the total possibilities for $\mathbf{D}$ over a finite field $\mathbb{F}_q$ is $q^{LK}$. Since $\mathbf{D} = \mathbf{C}\mathbf{A}$, the number of choices for the product $\mathbf{C}\mathbf{A}$ is also $q^{LK}$. However, the maximum possible choices for the product $\mathbf{C}\mathbf{A}$ is given by
\begin{align}
\frac{(\binom{L}{\Delta}q^{\Delta})^R(\binom{K}{M}q^M)^R}{(q-1)^R}.
\label{eq:count_conv}
\end{align}

The numerator in \eqref{eq:count_conv} represents the product of the number of distinct $\mathbf{C}$ and $\mathbf{A}$ matrices. The matrix $\mathbf{C}$ is of dimension $L \times R$ with entries from $\mathbb{F}_q$, and its columns are $\Delta-$sparse. Therefore, $\mathbf{C}$ has $(\binom{L}{\Delta}q^{\Delta})^R$ possibilities. Similarly, the total possibilities for $\mathbf{A}$ is $(\binom{K}{M}q^M)^R$ as $\mathbf{A} \in \mathbb{F}_q^{R \times K}$ and its rows are $M-$sparse. For a chosen $\mathbf{C}$ and $\mathbf{A}$, their scaled versions of the form $\mathbf{C}^{\prime}=\mathbf{C}{\Lambda}$ and $\mathbf{A}^{\prime}=\mathbf{\Lambda}^{-1}\mathbf{A}$, where $\Lambda$ is a $R \times R$ diagonal matrix defined as $\Lambda:=diag(\lambda_1,\lambda_2,\ldots,\lambda_R)$, also result in the same product. Note that the columns of $\mathbf{C}^{\prime}$ and rows of $\mathbf{A}^{\prime}$ are $\Delta-$sparse and $M-$sparse, respectively, as $\lambda_{r} \neq 0$ for all $r \in [R]$. Notice that the same product matrix is obtained if any other dense transformation matrix is used in place of $\Lambda$. However, the resulting $\mathbf{C}^{\prime}$ and $\mathbf{A}^{\prime}$ do not satisfy the required sparsity constraints. Therefore, the number of distinct product matrices $\mathbf{C}$ and $\mathbf{A}$ jointly generate is given by \eqref{eq:count_conv}.
Thus, we get the following relation
\begin{align}
q^{LK} \leq \frac{(\binom{L}{\Delta}q^{\Delta})^R(\binom{K}{M}q^M)^R}{(q-1)^R}.
\label{eq:conv_ineq}
\end{align}
Taking the logarithm of \eqref{eq:conv_ineq} and rearranging gives 
\begin{align}
R \geq \frac{LK}{\Delta + \log_q \binom{L}{\Delta} + M +\log_q\binom{K}{M}-\log_q(q-1)}.
\label{eq:conv_count}
\end{align}
Since $R^{*}(K,L,M,\Delta)$ is given by the infimum of achievable rates, we obtain \eqref{eq:conv_lb} from \eqref{eq:con_triv} and \eqref{eq:conv_count}.
\end{proof}

\begin{cor}
  When the field size $q \rightarrow \infty$, the lower bound on the optimal rate becomes 
  \begin{equation*}
  R^{*}(K,L,M, \Delta) \geq \frac{LK}{\Delta + M-1}.
  \end{equation*}
  \label{cor:largeq}
  \end{cor}
\begin{proof}
The proof follows directly from \eqref{eq:conv_lb}. As $q \rightarrow \infty$, the terms $\log_q\binom{L}{\Delta} \rightarrow 0$, $\log_q\binom{K}{M} \rightarrow 0$, and $\log_q(q-1) \rightarrow 1$.
\end{proof}

\begin{rem}
Corollary \ref{cor:largeq} implies that as the field size $q \rightarrow \infty$, the lower bound in \eqref{eq:conv_count}  matches the achievable rate in Theorem \ref{thm1} when $\Delta | L$ and $(\Delta+M-1)|K$.
\end{rem}

\begin{thm}[Optimality gap]
When $\Delta | L$, $(\Delta + M-1) | K$ and  $q \geq K$, the achievable rate in Theorem \ref{thm1} satisfies 
\begin{equation*}
  1\leq \frac{R_{\text{ach}}(K,L,M,\Delta)}{R^{*}(K,L,M,\Delta)} \leq 3.
\end{equation*}
\label{thm:optgap}
\end{thm}
\begin{proof}
   Consider the expression $R_{\text{aux}}=\frac{LK}{\Delta + M-1}$. Comparing $R_{\text{aux}}$ with the lower bound in \eqref{eq:conv_count}, we get
   \begin{align}
         \frac{R_{\text{aux}}}{R^{*}}  \leq 1 + \frac{1+\log_q\binom{L}{\Delta}+\log_q\binom{K}{M}-\log_q(q-1)}{\Delta +M-1}
         \label{eq:boundexp}
   \end{align}
   Using the fact that $\log_q(q-1) \leq 1$ for $q\geq 2$, and upper bounding the binomial coefficient by $\binom{n}{k} \leq n^k$, \eqref{eq:boundexp} becomes
   \begin{equation*}
         \frac{R_{\text{aux}}}{R^{*}}  \leq 1 + \frac{\Delta\log_q{L}+M\log_q{K}}{\Delta +M-1}.
 \end{equation*}
When $q \geq K$, we obtain 
\begin{equation}
\frac{R_{\text{aux}}}{R^*} \leq 3
\label{eq:con_aux}
\end{equation}
as $\log_qL \leq 1$ and $\log_qK \leq 1$. When $\Delta | L$ and $(\Delta +M-1) |K$, we have $R_{\text{ach}}(K,L,M,\Delta) = R_{\text{aux}}$. Thus, we get
\begin{equation*}
 \frac{R_{\text{ach}}(K,L,M,\Delta)}{R^{*}(K,L,M,\Delta)}  \leq 3
 \end{equation*}
 when $\Delta | L$,  $(\Delta +M-1)|K$,  and $q\geq K$. 
\end{proof}
Note that when conditions $\Delta |L$ and $(\Delta+M-1)|K$ do not hold, the rate is within a factor of 8 from the lower bound.
When the underlying field is $\mathbb{R}$, we derive a lower bound on the rate similar to Theorem \ref{thm:lb} using the degrees of freedom of $\mathbf{D}$ and the degrees of freedom of $\mathbf{C}$ and $\mathbf{A}$.
\begin{thm}
For a $(K,L,M,\Delta)$ linearly separable distributed computing problem defined over $\mathbb{R}$, the optimal rate satisfies
\begin{align}
 R^{*}(K,L,M,\Delta) \geq \max\left(L,\frac{LK}{\Delta+M-1}\right).
 \label{eq:conv_real}
\end{align}
\label{thm:conv_real}
\end{thm}
\begin{proof}
A demand matrix $\mathbf{D} \in \mathbb{R}^{L\times K}$ of rank $L$ has $LK$ degrees of freedom (DoF).\footnote{The degrees of freedom of a matrix is defined in \cite{AnC, Matman2}. The degrees of freedom of a matrix can be viewed as the number of independent parameters that can be varied in the matrix.} 
Since $\mathbf{D}=\mathbf{CA}$, we have $\text{DoF}(\mathbf{D}) = \text{DoF}(\mathbf{C}\mathbf{A})\leq \text{DoF}(\mathbf{C})+\text{DoF}(\mathbf{A})$. The matrix $\mathbf{C} \in \mathbb{R}^{L \times R}$ has $\Delta-$sparse columns, therefore, $\text{DoF}(\mathbf{C})=\Delta R$. Similarly, $\text{DoF}(\mathbf{A})=M R$ as the matrix $\mathbf{A} \in \mathbb{R}^{R \times K}$ has $M-$sparse rows. Thus, we get 
\begin{align*}
LK \leq \Delta R + MR.
\end{align*}
which can be tightened by considering the redundant DoF as done in the finite field case (Theorem \ref{thm:lb}). i.e., the scaled versions of $\mathbf{C}$ and $\mathbf{A}$ having the following form $\mathbf{C\Lambda}$ and $\mathbf{\Lambda^{-1}A}$, where $\Lambda \in \mathbb{R}^{R \times R}$ is a diagonal matrix of rank $R$, also result in the same product. Therefore, we get $\text{DoF}(\mathbf{D}) \leq \text{DoF}(\mathbf{C}) + \text{DoF}(\mathbf{A})-\text{DoF}(\mathbf{\Lambda})$. Notice that the $\text{DoF}(\mathbf{\Lambda})=R$ as $\Lambda \in \mathbb{R}^{R \times R}$ is a diagonal matrix of rank $R$. Thus, we have
\begin{align*}
LK \leq \Delta R+MR-R 
\end{align*}
which results in 
\begin{align*}
R \geq \frac{LK}{\Delta +M -1}.
\end{align*}
This completes the proof of Theorem \ref{thm:conv_real}.
\end{proof}

\begin{rem}
For the real field, the achievable rate in Theorem \ref{thm1} is optimal when $\Delta | L$ and $(\Delta+M-1)|K$; otherwise, it is within a factor of $4$ from the optimal rate.
\end{rem}


\bibliographystyle{IEEEtran}
\bibliography{Mybibliography}

\end{document}